\documentclass[runningheads]{llncs}

\usepackage{amssymb,amsmath}
\usepackage{graphicx,subfigure,tabularx}
\usepackage{enumerate}
\usepackage{hyperref}
\usepackage{color}
\usepackage{times}
\usepackage{paralist}
\usepackage{wrapfig}
\usepackage{verbatim}
\usepackage[boxed,vlined]{algorithm2e}
\DontPrintSemicolon
\SetArgSty{}

\usepackage[activate=normal]{pdfcprot}

\def\comment#1{}%
\def\withcomments{%
  \newcounter{mycommentcounter}%
   \def\comment##1{\refstepcounter{mycommentcounter}%
    \ifhmode%
     \unskip%
     {\dimen1=\baselineskip \divide\dimen1 by 2 %
       \raise\dimen1\llap{\tiny
	{-\themycommentcounter-}}}\fi%
     \marginpar[{\renewcommand{\baselinestretch}{0.8}%
       \hspace*{-2em}\begin{minipage}{1.3\marginparwidth}\footnotesize%
[\themycommentcounter]:%
\raggedright ##1\end{minipage}}]{\renewcommand{\baselinestretch}{0.8}%
       \begin{minipage}{1.3\marginparwidth}\footnotesize%
[\themycommentcounter]: \raggedright%
##1\end{minipage}}}%
  }

\newcommand{\comm}[1]{}

\newcommand{\bc}{\operatorname{bc}}
\newcommand{\des}{\operatorname{des}}
\newcommand{\gap}{\operatorname{gap}}
\newcommand{\lis}{\operatorname{lis}}
\newcommand{\bp}{\operatorname{bp}}
\newcommand{\bcalg}{\operatorname{bc}_{\operatorname{alg}}}
\newcommand{\bcopt}{\operatorname{OPT}}
\newcommand{\bciialg}{\overline{\operatorname{bc}}_{\operatorname{alg}}}
\newcommand{\bciialgt}{\overline{\operatorname{bc}}_{\operatorname{alg}t}}
\newcommand{\bciialgb}{\overline{\operatorname{bc}}_{\operatorname{alg}b}}
\newcommand{\bciiopt}{\overline{\operatorname{bc}}_{\operatorname{opt}}}
\newcommand{\gp}{\operatorname{gp}}
\newcommand{\gpinit}{\operatorname{gp}_{\operatorname{init}}}

\let\doendproof\endproof
\renewcommand\endproof{~\hfill\qed\doendproof}

\begin{document}
\date{}

\title{Ordering Metro Lines by Block Crossings}

\author{Martin Fink \inst1 \and Sergey Pupyrev \inst2}

\institute{
    Lehrstuhl f\"{u}r Informatik I, Universit\"{a}t
    W\"{u}rzburg, Germany. 
\and 
    Department of Computer Science, University of Arizona, USA.
}

\maketitle
\begin{abstract}
  A problem that arises in drawings of transportation networks is to
  minimize the number of crossings between different transportation
  lines. While this can be done efficiently under specific constraints,
  not all solutions are visually equivalent. We suggest
  merging crossings into \emph{block crossings}, that is, crossings
  of two neighboring groups of consecutive lines. Unfortunately,
  minimizing the total number of block
  crossings is NP-hard even for very simple graphs. We
  give approximation algorithms for special classes of graphs and
  an asymptotically worst-case optimal algorithm for block crossings
  on general graphs. That is, we bound the number of block crossings
  that our algorithm needs and construct worst-case instances on which the
  number of block crossings that is necessary
  in any solution is asymptotically the same as our bound.
\end{abstract}


\section{Introduction}
\label{sec:intro}
In many metro maps and transportation networks
some edges, that is, railway
track or road segments, are used by several \emph{lines}. Usually, to
visualize such networks, lines that share an edge are drawn individually
along the edge in distinct colors. Often, some lines must cross,
and it is desirable to draw the lines with few crossings. The
\emph{metro-line crossing minimization} problem has recently been
introduced~\cite{benkert07}. The goal
is to order the lines along each edge such that the number of
crossings is minimized. So far, the focus has been on the number of
crossings and not on their visualization, although two line orders
with the same crossing number may look quite differently;
see~Fig.~\ref{fig:metro}.

\begin{figure}[tb]
    ~\hfill\includegraphics[width=0.45\textwidth,page=1]{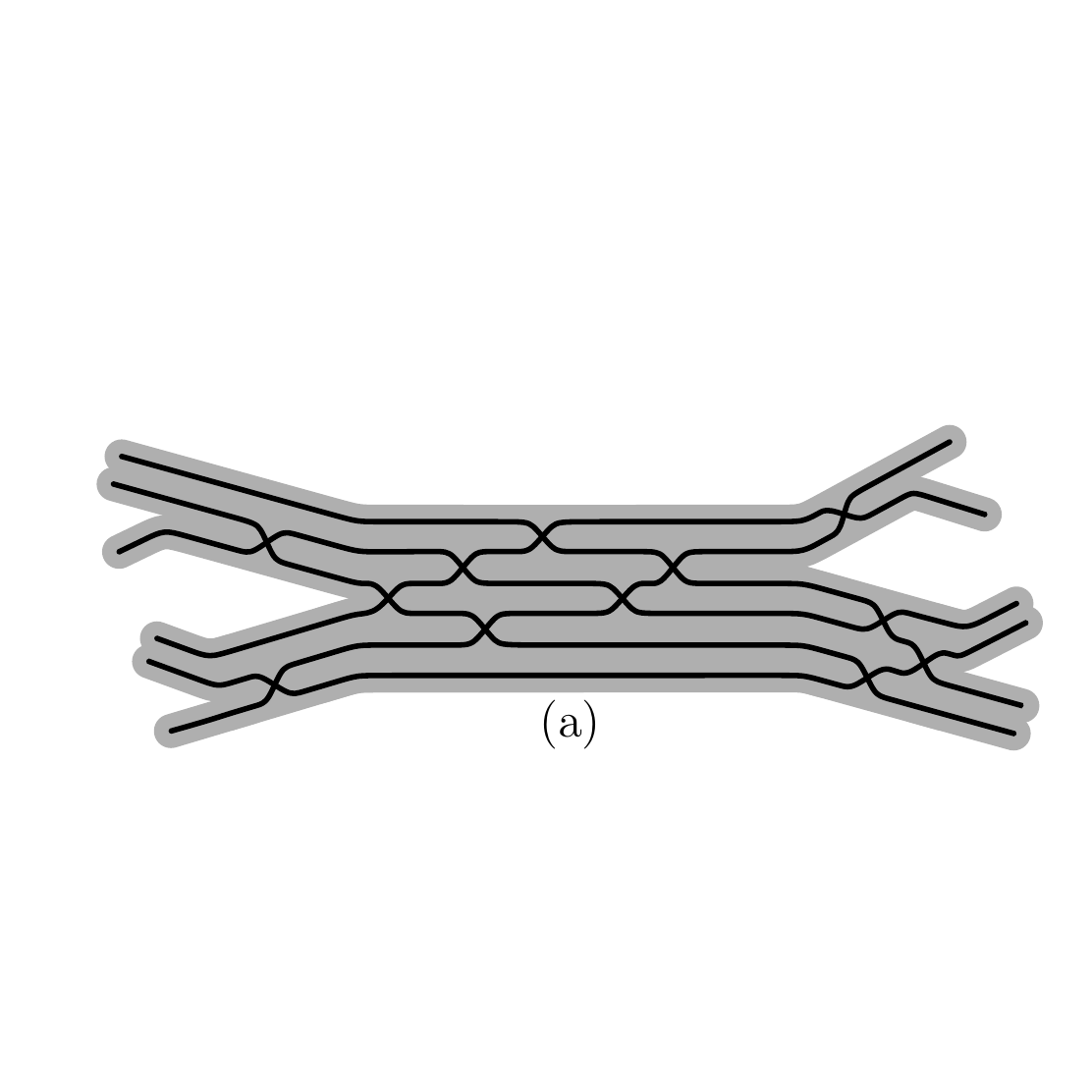}  \hfill
    \includegraphics[width=0.45\textwidth,page=2]{pics/metro0}\hfill~
  \caption{Optimal orderings of a metro network: (a)~12 pairwise
  crossings; (b)~3 block crossings.}
  \label{fig:metro}
\end{figure}

Our aim is to improve the readability of metro maps by computing line
orders that are aesthetically more pleasing. To this end, we merge
\emph{pairwise} crossings
into crossings of blocks of lines minimizing the number of \emph{block
crossings} in the map. Informally, a block crossing is an intersection
of two neighboring groups of consecutive lines
sharing the same edge; see~Fig.~\ref{fig:metro}(b). We consider two variants of the problem.
In the first variant, we want to find a line ordering with the minimum number
of block crossings. In the second variant, we want to minimize both
pairwise and block crossings.

\subsubsection{Motivation.}
Although we present our results in terms of the classic problem of visualizing
metro maps, crossing minimization between paths on an embedded
graph is used in various fields. In very-large-scale integration
(VLSI) chip layout, a wire diagram  should have few wire crossings~\cite{groeneveld89a}.
Another application
is the visualization of biochemical pathways~\cite{schreiber2002high}.
In graph drawing the number of edge crossings is
considered one of the most popular aesthetic criteria. Recently, a lot
of of research, both in
graph drawing and information visualization, is devoted to
\emph{edge bundling}. In this setting, some edges are drawn close
together---like metro lines---which emphasizes the structure of~the graph~\cite{pupyrev11}.
Block crossings can greatly improve the
readability of bundled graph drawings.

\subsubsection{Problem definition.}
The input consists of an embedded graph~$G=(V,E)$, and
a set $L=\{l_1,\ldots,l_{|L|}\}$
of simple paths in $G$. We call $G$ the \emph{underlying network} and the
paths \emph{lines}. The nodes of $G$ are
\emph{stations} and the endpoints $v_0,v_k$ of a line
$(v_0, \ldots,v_k)\in~L$  are \emph{terminals}.
For each edge $e = (u,v)\in E$, let $L_{e}$ be the set of lines
passing through~$e$. For $i \le j < k$, a \emph{block move} $(i,j,k)$
on the sequence $\pi=\left[\pi_1,\dots,\pi_n\right]$ of lines on $e$ is the exchange of two
consecutive blocks
$\pi_i,\dots,\pi_j$ and $\pi_{j+1},\dots,\pi_k$. We are interested in
\emph{line orders} $\pi^0(e),\dots,\pi^{t(e)}(e)$ on $e$, so that
$\pi^0(e)$ is the order of lines $L_e$ on $e$ close to $u$,
$\pi^{t(e)}(e)$ is the order close to $v$, and
each $\pi^i(e)$ is an ordering of $L_{e}$ so that $\pi^{i+1}(e)$ is
constructed from $\pi^i(e)$ by a block move.
We say that there are $t$ \emph{block crossings} on~$e$.

\begin{wrapfigure}[17]{r}{.293\textwidth}
  \centering
    \includegraphics[height=2.1cm,page=1]{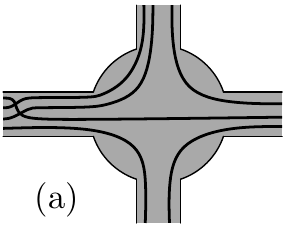}
    \hfill
    \includegraphics[height=2.1cm,page=2]{pics/consistency}
  \caption{Consistent line orders (a) without, (b) with an
  unavoidable vertex crossing.}
  \label{fig:consistency}
\end{wrapfigure}
Following previous work \cite{argyriou09,nollenburg09}, we use the
\emph{edge crossings} model,
that is, we do not hide crossings under station symbols if possible.
Two lines sharing at least one common edge
either do not cross or cross each other on an edge but never in a
node; see Fig.~\ref{fig:consistency}(a). For pairs of lines sharing
a vertex but no edges, crossings at the vertex are allowed and not
counted as they exist at this position
in any solution. We call them
\emph{unavoidable vertex crossings}; see
Fig.~\ref{fig:consistency}(b). If the line orders on the edges
incident to a vertex~$v$ produce only edge crossings and unavoidable vertex
crossings, we call them \emph{consistent} in $v$. Line orders for all edges are
consistent if they are consistent in all nodes.
More formally, we can check consistency of line orders in a vertex $v$ by
looking at each incident edge $e$. Close to $v$ the order of lines $L_e$ on
$e$ is fixed. The other edges $e_1, \ldots, e_k$ incident to
$v$ contain lines
of $L_e$. The combined order of $L_e$ on the edges $e_{1}, \ldots,
e_k$ must be the same as the order on $e$; otherwise, lines of
$L_e$ would cross in $v$.
Now, we can define the \emph{block crossing minimization} problem (BCM).

\begin{problem}[\textbf{BCM}]
Let $G=(V,E)$ be an embedded graph and let $L$ be a set of lines on~$G$.
For each edge $e \in E$, find line orders $\pi^0(e),\dots,\pi^{t(e)}(e)$ such
that the total number of block crossings, $\sum_{e\in E}t(e)$, is minimum
and the line orders are consistent.
\end{problem}

In this paper, we restrict our attention to instances with two
additional properties.
First, any line terminates at nodes of degree one and no two lines
terminate at the same node (\textbf{path terminal property}).
Second, the intersection of two lines, that is, the edges and vertices
they have in common, forms a path (\textbf{path
intersection property}). This includes the cases that the intersection
is empty or a single node.
If both properties hold, a pair of lines
either has to cross, that is, a crossing is
\emph{unavoidable}, or it can be kept crossing-free, that is, a crossing
is \emph{avoidable}. The orderings that are optimal with
respect to pairwise crossings are exactly the orderings that
contain just unavoidable crossings (Lemma 2 in~\cite{nollenburg09}); that is, any pair of lines
crosses at most once, in an equivalent formulation. As this is a very
reasonable condition also for block crossings, we use it to define the
\emph{monotone block crossing minimization} problem (MBCM) whose
feasible solutions must have the minimum number of pairwise
crossings.
\begin{problem}[\textbf{MBCM}] Given an instance of BCM, find a
  feasible solution that minimizes the number of block crossings
  subject to the constraint that no two lines cross twice.
\end{problem}
On some instances BCM does allow fewer crossings than
MBCM does; see Fig.~\ref{fig:bc_vs_mbc}.

\subsubsection{Our contribution.}
We introduce the new problems BCM and MBCM.
To the best of our knowledge, ordering lines by block
crossings is a new direction in graph drawing. So far BCM
has been investigated only for the case that the
\emph{skeleton}, that is, the graph without terminals, is a single
edge~\cite{bafna98}, while MBCM is a completely new problem.

We first analyze MBCM on a single edge (Sec.~\ref{sec:edge}),
exploiting, to some extent,
the similarities to \emph{sorting by transpositions}~\cite{bafna98}. Then, we use the
notion of \emph{good pairs} of lines, that is, lines that should be
neighbors, for developing an approximation algorithm for BCM on graphs
whose skeleton is a path (Sec.~\ref{sec:path}); we
properly define good pairs so that changes between adjacent edges are
taken into account. Yet, good pairs can not always be kept close;
we introduce a good strategy for breaking pairs when needed.

Unfortunately, the approximation algorithm does not generalize to
trees. We do, however, develop a worst-case optimal algorithm for
trees (Sec.~\ref{sec:tree}). It needs $2|L|-3$ block crossings and
there are instances in which this number of block crossings is
necessary in any solution. We then use our algorithm
for obtaining approximate solutions for MBCM on the special class of
\emph{upward trees}.

\begin{table}[t]
    \centering
    \begin{tabular}{|l|c|c|c|c|}
        \hline
        \centering graph class &
        \multicolumn{2}{c|}{BCM} &
        \multicolumn{2}{c|}{MBCM} \\
        \hline
        \hline
        single edge         & $11/8$-approximation & \cite{elias06} &
            3-approximation & Sec.~\ref{sec:edge} \\
        path                & 3-approximation  & Sec.~\ref{sec:path}
                            & 3-approximation & Sec.~\ref{sec:pathr-appendix} \\
        tree                & $\le 2|L|-3$ crossings &
        Sec.~\ref{sec:tree} & $\le 2|L|-3$ crossings & Sec.~\ref{sec:tree}\\
        upward tree         & \centering ---  & & 6-approximation &
        Sec.~\ref{sec:utree}\\
        general graph       & $O(|L|\sqrt{|E|})$ crossings  &
       Sec.~\ref{sec:graph} & $O(|L|\sqrt{|E|})$ crossings & Sec.~\ref{sec:graph}\\
        \hline
    \end{tabular}
    \medskip
   \caption{Overview of our results for BCM and MBCM.}
    \label{table:or}
\end{table}

As our main result, we develop an algorithm for obtaining a solution
for (M)BCM on general graphs (Sec.~\ref{sec:graph}). We show that it
uses only monotone block moves and analyze the upper bound
on the number of block crossings. While the algorithm itself is simple and
easy to implement, proving the upper bound is non-trivial.
Next, we show that the bound is tight; we use
a result from projective geometry for constructing worst-case examples
in which any feasible solution contains many block crossings. Hence,
our algorithm is asymptotically worst-case optimal.
Table~\ref{table:or} summarizes our results.

\subsubsection{Related work.}
Line crossing problems in transportation networks were initiated by
Benkert et~al.~\cite{benkert07}, who considered the problem of
\emph{metro-line crossing minimization} (MLCM) on a single
edge. MLCM in its general model
is challenging; its complexity is open and no efficient algorithms
are known for the case of two or more edges.
Bekos~et~al.~\cite{bekos08} addressed the problem on paths and trees.
They also proved that a variant in which all lines must be placed outermost in their
terminals is NP-hard. Subsequently, Argyriou~et~al.~\cite{argyriou09} and
N\"ol\-len\-burg~\cite{nollenburg09} devised polynomial-time algorithms for
general graphs with the  path terminal property.
Pupyrev~et~al.~\cite{pupyrev11} studied MLCM in the
context of edge bundling. They suggested a linear-time algorithm
for MLCM on instances with the path terminal property. All
these works are dedicated to pairwise crossings; the optimization
criterion being the number of crossing pairs of lines.

A closely related problem arises in VLSI design, where the goal is to
minimize intersections between nets (physical wires)~\cite{groeneveld89a,mareksadowska95}.
Net patterns with fewer crossings most likely have better electrical
characteristics and require less wiring area; hence, it is an important
optimization criterion in circuit board design.
Marek-Sadowska and Sarraf\-za\-deh~\cite{mareksadowska95} considered not
only minimizing the number of crossings, but also
suggested distributing the crossings among circuit regions in
order to simplify net routing.

BCM on a \emph{single} edge is equivalent to the problem of sorting a
permutation by block moves, which is well studied in computational
biology for DNA sequences; it is known as \emph{sorting by
transpositions}~\cite{bafna98,christie01}. The task is to find the
shortest sequence of block moves transforming a given permutation into
the identity permutation. The complexity of the
problem was open for a long time; only recently it has been shown to be
NP-hard~\cite{bulteau11}. The currently best known algorithm has an
approximation ratio of $11/8$~\cite{elias06}.
The proof of correctness of that algorithm is based on a computer
analysis which verifies more than $80,000$ configurations.
To the best of our knowledge, no tight upper bound is know for the problem.
There are several variants of sorting by transpositions; see the
survey of Fertin et al.~\cite{cogr}. For instance,
Vergara~et~al.~\cite{Heath98} used \emph{correcting short block moves}
to sort a permutation. In our terminology, these are monotone moves
such that the combined length of exchanged blocks does not exceed
three. Hence, their problem is a restricted variant of MBCM on a
single edge; its complexity is unknown.


\section{Block Crossings on a Single Edge}
\label{sec:edge}

First, we restrict our attention to networks consisting of a single
edge with multiple lines passing through it. BCM then can be
reformulated as follows. Given two permutations $\pi$ and $\tau$
(determined by the order of terminals on both sides of the edge), find
the shortest sequence of block moves transforming $\pi$ into $\tau$.
By relabeling we can assume that $\tau$ is the identity permutation,
and the goal is to sort $\pi$. This problem is known as
\emph{sorting by transpositions} \cite{bafna98}. We
concentrate on the new problem of
sorting with monotone block moves; that means that the relative order
of any pair of elements changes at most once. The problems are not
equivalent; see~Fig.\ref{fig:bc_vs_mbc} for an example where
non-monotonicity allows fewer crossings.
In what follows, we give lower and upper bounds on the number of block
crossings for MBCM on a single edge. Additionally, we present a simple
3-approximation algorithm for the problem.

\begin{wrapfigure}[13]{r}{.31\textwidth}
  \centering
    \includegraphics[height=2.5cm]{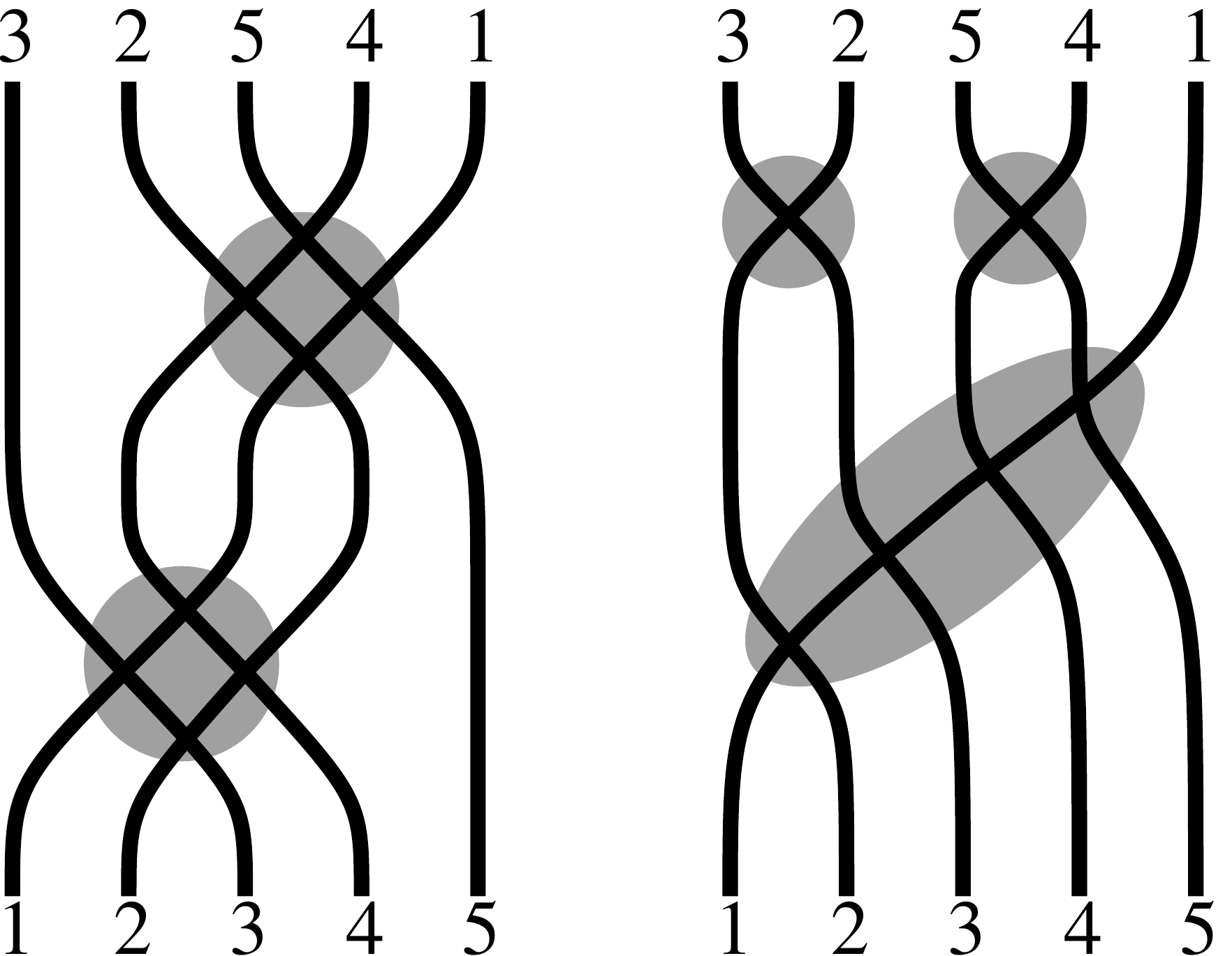}
  \caption{Permutation [3 2 5 4 1] is sorted with 2 block
      moves and 3 monotone block moves.}
  \label{fig:bc_vs_mbc}
\end{wrapfigure}
We first introduce some terminology following the one from previous works where possible. Let $\pi=[\pi_1,\dots,\pi_n]$ be
a permutation of $n$ elements. For convenience, we assume there are
extra elements $\pi_0 = 0$ and $\pi_{n+1} = n+1$ at the beginning of
the permutation and at the end, respectively.
A \emph{block} in $\pi$ is a sequence of consecutive elements $\pi_i,\dots,\pi_j$
with $i \le j$. A \emph{block move} $(i, j, k)$ with $i\le j < k$ on $\pi$ maps $[\dots
\pi_{i-1} \pi_{i} \ldots \pi_{j}\pi_{j+1} \dots \pi_{k}
\pi_{k+1}\dots]$ to $[\dots \pi_{i-1} \pi_{j+1} \ldots \pi_{k}\pi_{i}
  \dots \pi_{j} \pi_{k+1}\dots]$. We say that a block move $(i,j,k)$
is \emph{monotone} if $\pi_q > \pi_r$ for all $i \le q \le j < r \le k$.
We denote the minimum number of monotone block moves needed to sort $\pi$ by $\bc(\pi)$.
An ordered pair $(\pi_i, \pi_{i+1})$ is a \emph{good pair}
if $\pi_{i+1} = \pi_i+1$, and a \emph{breakpoint} otherwise.
Intuitively, sorting $\pi$ is a process of creating good pairs (or
destroying breakpoints) by block moves.
A permutation is \emph{simple} if it has no good pairs. Any
permutation can be uniquely simplified---by glueing good pairs
together and
relabeling---without affecting its distance to the identity
permutation~\cite{christie01}.
A breakpoint $(\pi_i, \pi_{i+1})$ is a \emph{descent} if $\pi_i > \pi_{i+1}$, and
a \emph{gap} otherwise. We use $\bp(\pi)$, $\des(\pi)$, and $\gap(\pi)$ to denote the number
of breakpoints, descents, and gaps in $\pi$.
The \emph{inverse} of a permutation $\pi$ is the permutation $\pi^{-1}$ in which each
element and the index of its position are exchanged, that is,
$\pi_{\pi_i}^{-1} = i$ for $1 \le i \le n$.
A descent in $\pi^{-1}$, that is, a pair of elements $\pi_i = \pi_j+1$
with $i<j$, is called an \emph{inverse descent} in $\pi$.
Analogously, an \emph{inverse gap} is a pair of elements $\pi_i =
\pi_j+1$ with $i>j+1$.
Now, we give lower and upper bounds for MBCM, that is, on~$\bc(\pi)$.

\subsubsection{A lower bound.}
It is easy to see that a block move affects three pairs of adjacent elements.
Therefore the number of breakpoints can be reduced by at most three
in a move. As only the identity permutation has no breakpoints, this implies
$\bc(\pi) \ge \bp(\pi)/3$ for a simple permutation $\pi$~\cite{bafna98}.
The following observations yield better lower bounds.

\begin{lemma}
In a monotone block move, the number of descents in a permutation
decreases by at most one, and the number of gaps decreases by at
most two.
\end{lemma}
\begin{proof}
Consider a monotone move
$[\dots a\:b \dots c\:d \dots e\:f \dots] \Rightarrow [\dots a\:d \dots e\:b
\dots c\:f \dots]$;
it affects three adjacencies. Suppose a descent is destroyed between
$a$ and $b$, that is,
$a > b$ and $a < d$. Then, $b < d$ which contradicts monotonicity.
Similarly, no descent can be destroyed between $e$ and $f$.
On the other hand, since $c>d$, no gap can be destroyed
between $c$ and $d$.
\end{proof}

A similar claim holds for the inverse descents and gaps.

\begin{lemma}
In a monotone block move, the number of inverse descents decreases by
at most one, and the number of inverse gaps decreases by at most two.
\end{lemma}
\begin{proof}
Consider a monotone exchange of blocks $\pi_i,\dots,\pi_j$ and $\pi_{j+1},\dots,\pi_k$.
Note that inverse descents can only be destroyed between elements
$\pi_q$ $(i \le q \le j)$ and $\pi_r$ $(j+1 \le r \le k)$. Suppose that
the move destroys two inverse descents such that the first block
contains elements $x+1$ and $y+1$, and
the second block contains $x$ and $y$. Since the block move is
monotone, $y+1 > x$ and $x+1 > y$, which means that $x=y$.

On the other hand, there cannot be inverse gaps between elements $\pi_q$ $(i \le q \le j)$ and
$\pi_r (j+1 \le r \le k)$. Therefore, there are only two
possible inverse gaps between $\pi_{i-1}$ and
$\pi_r (j < r \le k)$, and between $\pi_q (i \le q \le j)$ and $\pi_{k+1}$.
\end{proof}

Combining the lemmas, we obtain the following result.

\begin{theorem}
A lower bound on the number of monotone block moves needed to sort a permutation is
$\bc(\pi) \ge \max(\bp(\pi)/3, \des(\pi), \gap(\pi)/2, \des(\pi^{-1}), \gap(\pi^{-1})/2)$.
\end{theorem}

\subsubsection{An upper bound.}
We suggest the following algorithm for sorting a simple
permutation $\pi$: In each step find the smallest $i$ such that
$\pi_i\neq i$ and move element $i$ to position $i$, that is, exchange block
$\pi_i,\dots,\pi_{k-1}$ and $\pi_k$, where $\pi_k=i$.
Clearly, the step destroys at least one
breakpoint. Therefore $\bc(\pi) \le \bp(\pi)$ and the algorithm
yields a 3-approximation.

\begin{theorem}
There exists an $O(n^{2})$-time 3-approximation algorithm for MBCM on a single edge.
  \label{thm:single-edge-mbcm-approx}
\end{theorem}

To construct a better upper bound, we first consider a constrained
sorting problem in which at least one of the moved blocks has unit
size; that is, we allow only  block moves of types $(i,i,k)$ and
$(i,k-1,k)$. Let $\bc^1(\pi)$ be
the minimum number of such block moves needed to sort $\pi$. We show
how to compute $\bc^1(\pi)$ exactly.
An \emph{increasing subsequence} of $\pi$ is a sequence
$\pi_{l_1}, \pi_{l_2}, \dots$ such that $\pi_{l_1} < \pi_{l_2} <
\dots$ and $l_1 < l_2 < \dots$.
Let $\lis(\pi)$ be the size of the longest increasing subsequence of $\pi$.

\begin{lemma}
$\bc^1(\pi) = n - \lis(\pi)$.
\end{lemma}

\begin{proof}
$\bc^1(\pi) \ge n - \lis(\pi)$. Consider a monotone move
$\sigma=(i, i, k)$ in $\pi$ ($\sigma = (i,k-1,k)$ is symmetric).
Let $\tilde{\pi} = [\pi_1,\dots,\pi_{i-1},\pi_{i+1},\dots,\pi_n]$
be the permutation $\pi$ without element $\pi_i$. Clearly,
$\lis(\tilde{\pi}) \le \lis(\pi)$. If we apply $\sigma$, the resulting
permutation
$\sigma\pi = [\pi_1, \dots, \pi_{i-1}, \pi_{i+1}, \dots, \pi_k, \pi_i,
\pi_{k+1}, \dots, \pi_n]$
has one extra element compared to $\tilde{\pi}$, and, therefore,
$\lis(\sigma\pi) \le \lis(\tilde{\pi}) + 1$. Hence,
$\lis(\sigma\pi) \le \lis(\pi) + 1$, that is, the length of the longest increasing
subsequence cannot increase by more than one in a move. The inequality
follows since $\lis(\tau) = n$ for the identity permutation $\tau$.

$\bc^1(\pi) \le n - \lis(\pi)$. Let $S = [\dots s_1 \dots s_2
\dots s_{\lis} \dots]$ be a fixed longest increasing
subsequence in $\pi$.
We show how to choose a move that increases the length of $S$.
Let $\pi_i \notin S$ be the rightmost
element (that is, $i$ is maximum) lying between elements $s_j$ and $s_{j+1}$ of
$S$ so that $\pi_i > s_{j+1}$. We move $\pi_i$ rightwards to its proper
position $p_i$ inside $S$. This is a monotone move, as $\pi_i$ was chosen
rightmost.  If no such element $p_i$ exists, we symmetrically choose
the leftmost $p_i$ with $p_i < s_j$ and bring it into its proper
position in $S$. In both cases $S$ grows.
\end{proof}


\begin{corollary}
Any permutation can be sorted by $n - \lis(\pi)$ monotone block moves.
\end{corollary}

\section{Block Crossings on a Path}
\label{sec:path}

Now we consider an embedded graph $G = (V,E)$ consisting of a path $P
= (V_P,E_P)$ with attached terminals. In every node $v\in V_P$ the clockwise
order of terminals adjacent to $v$ is given, and we assume
the path is oriented from left to right. We say that a line $l$ starts
at $v$ if $v$ is the leftmost vertex on $P$ that
lies on~$l$ and ends at its rightmost vertex of the path.
As we consider only
crossings of lines sharing an edge, we assume that the terminals
connected to any path node $v$ are
in such an order that first lines end at $v$ and then lines start at
$v$; see Fig.~\ref{fig:alg-path-edge-step}.

\begin{wrapfigure}[9]{r}{.29\textwidth}
  \centering
  \includegraphics{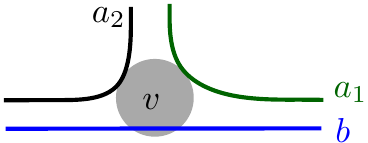}
  \caption{Inheritance of a good pair above node $v$.}
  \label{fig:inheritance}
\end{wrapfigure}
We suggest a 3-approximation algorithm for BCM. Similar
to the single edge case, the basic idea of the algorithms is to
consider
good pairs of lines. A \emph{good pair} is, intuitively, an ordered pair of lines that will be
adjacent---in this order---in any feasible solution when one of the lines ends.
We argue that our algorithm creates at least one additional good pair
per block crossing, while even the optimum creates at most three new
good pairs per crossing.  To describe our algorithm we first define
good pairs.

\begin{definition}[Good pair]~
  \label{def:good_pair}
  \begin{compactenum}[(i)]
    \item If two lines $a$ and $b$ end on the same node, and
      $a$ and $b$ are consecutive in clockwise order, then
      $(a,b)$ is a \emph{good pair} (as it is in the case of a single edge
      in Sec.~\ref{sec:edge}).

    \item Let $v$ be a node with edges $(u,v)$ and $(v,w)$ on $P$,
      let $a_{1}$ be the first line starting on $v$ above $P$,
      and let $a_{2}$ be the last line ending on $v$ above $P$ as in
      Fig.~\ref{fig:inheritance}.  If $(a_{1}, b)$ is a good pair,
      then $(a_{2}, b)$ also is a good pair.  We say that $(a_{2},b)$
      is \emph{inherited} from $(a_{1},b)$, and identify $(a_{1},b)$
      with $(a_{2},b)$, which is possible as $a_1$ and $a_2$ do not
      share an edge. Analogously, there is inheritance for lines
      starting/ending below~$P$.
  \end{compactenum}
\end{definition}

As a preprocessing step, we add a virtual line $t_{e}$ ($b_{e}$) for each edge
$e\in E_P$. The line $t_{e}$ ($b_{e}$) is the last line starting before $e$,
and the first line ending after $e$ to the top (bottom). Although
virtual lines are never moved, $t_{e}$ ($b_{e}$) does participate in
good pairs, which model the fact that the first lines
ending after an edge should be brought to the top (bottom).

There are important properties of good pairs.

\begin{lemma}
On an edge $e \in E_P$ there is, for each line $l$, at most one
good pair $(l',l)$ and at most one good pair $(l,l'')$.
\end{lemma}
\begin{proof}
Let $e$ be the rightmost edge where there is a line $l$ that
violates the property, that is, there are two good pairs $(l',l)$
and $(l'',l)$ (symmetrically for $l$ on the top).
If $l$ ends after $e$ there clearly can be at most one of
these good pairs. Suppose that $l$ also exists one the edge
$e'$ right of $e$. If both $l'$ and $l''$ exist on
$e'$, we would already have a counterexample on $e'$. Hence, at
least one of the lines ends after $e$, that is, at least one of the
good pairs results from inheritance after the edge $e$. On the
other hand, this can only be the case for one of both, suppose
for $(l',l)$. There has to be another good pair $(l''',l)$ on
$e'$, a contradiction to the choice of $e$.
\end{proof}

\begin{lemma}
If $e \in E_P$ is the last edge before line $l$ ends to the top
		(bottom), then there exists a line $l'$ ($l''$) on $e$ that forms a good
		pair $(l',l)$ ($(l,l'')$) with $l$.
\end{lemma}
\begin{proof}
We suppose that $l$ ends to the top; the other case is analogous. Let
$e = (u,v)$. We consider the clockwise order of lines ending around
$v$. If there is a predecessor $l'$ of $l$, then by case (i) of the
definition $(l',l)$ is a good pair. Otherwise, $l$ is the topmost line
ending at $v$. The virtual line $t_e$ that we added is its predecessor,
and $(t_i,l)$ is a good pair.
\end{proof}

In what follows, we say that a solution (or algorithm) creates a good
pair in a block crossing if the two lines of the good pair are brought
together in the right order by that block crossing; analogously, we
speak of breaking good pairs.
It is easy to see that any solution, especially an optimal one, has to create all good
pairs, and a block crossing can create at most three new pairs.
There are only two possible ways for creating a good pair $(a,b)$:
(i) $a$ and $b$ start at the same node consecutively in the right
order, that is, they form an \emph{initial good pair}, or (ii)
a block crossing brings $a$ and $b$ together.
Similarly, good pairs can only be destroyed by crossings before both lines end.

\begin{lemma}
There are only two possibilities to create a good pair
    $(a,b)$:
    \begin{enumerate}[(i)]
      \item $a$ and $b$ start at the same node consecutively in the
        right order.
      \item A block crossing brings $a$ and $b$ together.
    \end{enumerate}
\end{lemma}
\begin{proof}
      During $a$ and $b$ exist, the good pair $(a,b)$ can only be
      created by block crossings because either $a$ and $b$ have to
      cross each other or lines between $a$ and $b$ have to leave the path.
      Hence, $(a,b)$ can only be created without a block crossing at
      the moment when the last of the two lines, say $a$, starts at a
      node $v$. In this case $a$ has to be the first line starting at
      $v$ on the top of $P$. This implies that by inheritance
      there is a good pair $(c,b)$, where $c$ is the last line ending
      at $v$ to the top. It follows that the good pair $(c,b)$, which
      is identical to $(a,b)$, existed before $v$.
\end{proof}

\begin{figure}[t]
  \begin{minipage}[b]{0.47\textwidth}
    \centering
    \includegraphics[scale=.96]{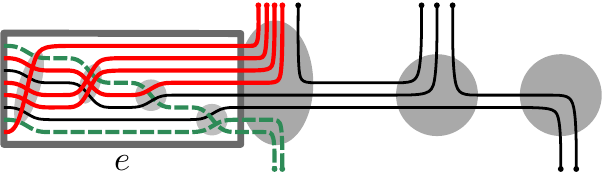}
    \caption{Ordering the lines on edge $e$ in a step of the algorithm.}
    \label{fig:alg-path-edge-step}
  \end{minipage}
  \hfill
  \begin{minipage}[b]{0.51\textwidth}
    \centering
    \includegraphics[height=1.8cm]{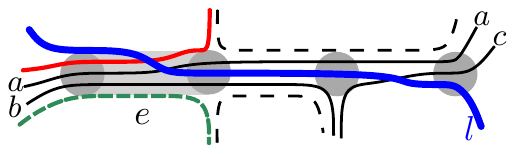}
    \caption{The (necessary) insertion of line $l$ forces
    breaking the good pair $(a,b)$ ($\equiv(a,c)$) on edge~$e$.}
    \label{fig:2-sided-destruction}
  \end{minipage}
\end{figure}

Using good pairs, we formulate our algorithm as follows; see
Fig.~\ref{fig:alg-path-edge-step} for an example.

We follow $P$ from left to right. On an edge $e = (u,v)$ there are red
lines that end at $v$ to the top, green lines that end at $v$ to the
bottom, and black lines that continue on the
next edge. We bring the red lines in the right order to the top by
moving them upwards. Doing so, we keep existing good pairs together.
If a line is to be moved, we consider the lines below it consecutively.
As long as the current line forms a good pair with the next line, we
extend the block that will be moved. We stop at the first line that
does not form a good pair with its successor. Finally we move the whole block of
lines linked by good pairs in one block move to the top.
Next, we bring the green lines in the right order to the bottom, again
keeping existing good pairs together.
There is an exception, where one good pair on $e$ cannot be
kept together. If the moved block is
a sequence of lines containing both red and
green lines, and possibly some---but not all---black lines, then it
has to be broken; see Fig.~\ref{fig:2-sided-destruction}.
Note that this can only happen in the last move on an edge.
There are two cases:
\begin{asparaenum}[(i)]
  \item A good pair in the sequence contains a black line and
    has been created by the algorithm previously. We break the
    sequence at this good pair.
	\item All pairs with a black line are initial
		good pairs, that is, were not created by a crossing. We break at the pair
		that ends last of these. Inheritance is also considered,
    that is, a good pair ends only when the last of the pairs that are
    linked by inheritance ends.
\end{asparaenum}

After an edge has been processed, the lines ending to the top and to
the bottom are on their respective side in the right relative order.
Hence, our algorithm produces a feasible solution. We
show that it produces a $3$-approximation for the number of block
crossings. A key property is that our strategy for case~(ii)
is optimal.

\begin{theorem}
  Let $\bcalg$ and $\bcopt$ be the number of block crossings created by the
	algorithm and an optimal solution, respectively. Then, $\bcalg \le 3 \bcopt$.
  \label{thm:3-approx-path}
\end{theorem}
\begin{proof}
  Normal block crossings, not breaking a good pair in the algorithm,
  always increase the number
  of good pairs. If we have a block crossing that breaks a good pair in a
  sequence as in case~(i) then there has been a block crossing
  that created the good pair previously as a side effect, that is, there
  was an additional (red or green) good pair whose creation caused
  that block crossing. Hence, we can say that the destroyed good pair
  did not exist previously and still have at least one new good pair
  per block crossing.

  If we are in case~(ii), that is, all good pairs in the sequence are
  initial good pairs; see Fig.~\ref{fig:2-sided-destruction}, then
  they also initially existed in the optimal solution. It is
  not possible to keep all those good pairs because
  the remaining black lines have to be somewhere between the block of
  red lines and the block of green lines. Hence, even the optimal
  solution has to break one of these good pairs on this edge or
  previously.

  Let $\bciialg$, $\bciiopt$ be the number of broken good pairs due to
  case (ii) in the algorithm and the optimal solution. In a crossing in
	which the algorithm breaks such a good pair the number of good pairs
  stays the same as one good pair is destroyed and another created. On
  the other hand, in a crossing that breaks a good pair, the number of
  good pairs can be increased by at most two even in the optimal
  solution (actually, it is not hard to see that it cannot be increased at
  all). Let $\gp$ be the total number of good pairs and let $\gpinit$
  be the number of initial good pairs. Note that according to
  Def.~\ref{def:good_pair} good pairs resulting from inheritance
  are not counted separately for $\gp$ as the are identified with another good
  pair. We get $\gp \ge \bcalg - \bciialg + \gpinit \text{~and~}
  \gp \le 3 \cdot \bcopt - \bciiopt + \gpinit$.
  Hence, $\bcalg \le 3 \bcopt + (\bciialg - \bciiopt)$ combining
  both estimates.

  To prove the approximation factor 3 we only have to show that $\bciialg \le
  \bciiopt$. First, note that the edges where good pairs of case~(ii)
  are destroyed, are exactly the edges where such a sequence
  of initial good pairs exists; that is, the edges are independent of
  any algorithm or solution. We show that, among these
  edges, our strategy ensures that the smallest
  number of pairs is destroyed, and pairs that are destroyed
  once are reused as often as possible for breaking a sequence of
  initial good pairs. To this end, let $e_{1}', \dots,
  e_{\bciialg}'$ be the sequence of edges, where the algorithm
  destroys a new good pair of type~(ii), that is, an initial good pair
  that has never been destroyed before. We follow the sequence and
  argue that the optimal solution destroys a new
  pair for each of these edges. Otherwise, there is a pair $e', e''$
  of edges in the sequence, where the
  optimal solution uses the same good pair $p$ on both edges. Let
  $p',p''$ be the pairs used by the algorithm on $e', e''$ for
  breaking a sequence of initial good pairs.
  As $p'$ was preferred by the algorithm over $p$, we know that $p'$
  still exists on $e''$.
  As it is in a sequence with $p$, the algorithm does, therefore,
  still use $p'$ on $e''$, a contradiction completing the proof.
\end{proof}

The algorithm needs $O(|L|(|L| + |E_{P}|))$ time. Note that it does
normally not produce orderings with monotone block crossings. It can,
however, be turned into a 3-appro\-xi\-ma\-tion algorithm for MBCM. To this end,
the definition of inheritance of good pairs, as well as the step of destroying good
pairs has to be adjusted, and the analysis has to be improved.

\subsubsection{Monotone Block Crossings on Paths.}
\label{sec:pathr-appendix}
We want to modify our algorithm so that
it produces monotone block crossings which are a $3$-approximation for
the minimum number of monotone block crossings in $O(|L|(|L| +
|E_P|))$ time. To this end, we first
have to modify our definition of inheritance of good pairs such that
we use only monotone block moves.
More specifically, we prevent inheritance in some situations in which
keeping a pair of lines together close to a vertex is not possible
without having a forbidden crossing. We concentrate on inheritance
with lines ending to the top; the other case is analogue.

\begin{wrapfigure}{r}{.29\textwidth}
  \centering
  \includegraphics{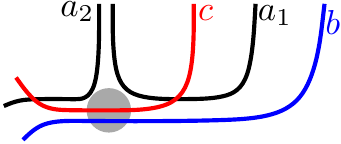}
  \caption{Line $c$ prevents that
  $(a_{2},b)$ inherits from $(a_{1},b)$.}
  \label{fig:monotone-no-inheritance-app}
\end{wrapfigure}
Suppose we have a situation as shown in
Fig.~\ref{fig:monotone-no-inheritance-app}. Line $c$ must not cross
$b$. On the other hand it has to be below $a_{2}$ near node $v$ and
separate $a_{2}$ and $b$ there. Hence, bringing or keeping $a_{2}$ and
$b$ together is of no value, as they have to be separated in any
solution. We say that line $b$ is \emph{inheritance-preventing} at
node~$v$.

One part of our algorithm still needs to be changed in order to ensure
monotonicity of the crossings. A block move including black lines
could result in a forbidden crossing. We focus on the case, where black
lines are moved together with red lines. This can only occur once per
edge. Let $r = b_0, b_{1}, \ldots, b_k$ be the sequence of good pairs from
the bottommost red line $r=b_0$ on. If there is some line $l$ above
the block that must not be crossed by a line $b_{i}$ of the block,
then we have to break the sequence. We consider such a case in which
$i$ is minimal.  Hence, we have to break a good pair out of $(r,
b_{1}), (b_{1}, b_{2}), \ldots, (b_{i-1}, b_{i})$. Similar to case
(i) in the algorithm of the previous section, we break a pair of this
sequence that is not initial. Otherwise (case (ii)), we choose the pair
$(b_{j-1}, b_{j})$ with $j \le i$ minimal such that the end node of
$b_{j}$ is below the path, and break the sequence there. Note that
line $l$ must end below the path, otherwise it would prevent
inheritance of at least one of the good pairs in the sequence. Hence,
also $b_{i}$ ends below the path, and $b_{j}$ is well-defined.

It is easy to see that our modified algorithm still produces a feasible
ordering. We now show that it is also monotone.

\begin{theorem}
  The algorithm produces an ordering with monotone block crossings.
  \label{theorem:algmonotone}
\end{theorem}
\begin{proof}
  We want too see that each block crossing is monotone, that is, a pair
  of lines that cross in a block crossing is in the wrong order before
  the crossing. Monotonicity of the whole solution then follows.
  We consider moves, where blocks of lines are brought to the top;
  the other case is analogue.

  Suppose a red line $r$ is brought to the top. As all red lines that have
  to leave above $r$ are brought to the top before, $r$ crosses only
  lines that leave below it, that is, lines that have to be crossed by
  $r$. If a black line $l$ is brought to the top, then it is moved together in a
  block that contains a sequence of good pairs from the bottommost red
  line $r'$ to $l$. Suppose $l$ crosses a line $c$ that should not be
  crossed by $l$. Line $c$ cannot be red because all red lines that
  are not in the block that is moved at the moment have been brought
  to the top before. It follows that $r'$ has to cross $c$. Hence, we
  can find a good pair $(a,b)$ in the sequence from $r'$ to $l$ such
  that $a$ has to cross $c$ but $b$ must not cross $c$.
  In this case, the algorithm will break at least one good pair
  between $r'$ and $b$. It follows that $c$ does not cross $l$, a
  contradiction.
\end{proof}

\begin{theorem}
  Let $\bcalg$ be the number of block crossings created by the
  algorithm and let $\bcopt$ be the number of block crossings of an
  optimal solution using only monotone block moves. It holds that
  $\bcalg \le 3 \bcopt$.
\end{theorem}
\begin{proof}
  As for non-monotone block crossings, all block crossings increase
  the number of good pairs, with the exception of breaking a sequence
  of initial good pairs in case (ii). Again, also the optimal solution
  has to have crossings, where such sequences are broken. In such
  a crossing, the two lines of the destroyed pair lose their partner.
  Hence, there is only one good pair after the crossing, and the
  number of good pairs does not change at all.
  Let $\bciialgt$ be the number of splits for case~(ii) where the
  block move brings lines to the top, and let $\bciialgb$ be the number of
  such splits where the move brings lines to the bottom.
  We get
  \begin{eqnarray*}
    \bcalg &\le& 3 \cdot \bcopt + (\bciialg - 3 \cdot \bciiopt)\\
    &\le& 3 \cdot \bcopt + (\bciialgt - \bciiopt) + (\bciialgb - \bciiopt)
  \end{eqnarray*}
  To complete the proof, we show $\bciialgt \le
  \bciiopt$, and symmetrically $\bciialgb \le \bciiopt$.

  Let $e_{1}', \dots, e_{\bciialgt}'$ be the sequence of edges, where
  the algorithm uses a new good pair, as a breakpoint for a sequence of
  type~(ii) when lines leave to the top, that is, a good pair that has not
  been destroyed before. Again, we argue that even the optimal
  solution has to use a different breakpoint pair for each of these
  edges. Otherwise, there would be a pair $e', e''$ of
  edges in this sequence, where the optimal solution uses the same
  good pair $p$ on both edges. Let $p'$ and $p''$ be the two good pairs used
  by the algorithm on $e'$ and $e''$, respectively. Let $p' = (l',l'')$. We
  know that
  $l'$ leaves the path to the top and $l''$ leaves to the bottom.
  Because all lines in the sequences on $e'$ and $e''$ stay parallel,
  we know that lines above $l'$ leave to the top, and lines below
  $l''$ leave to the bottom. Especially $p'$ still exists on
  $e''$, as $p$ stays parallel and also still exists.

  As in the description of the algorithm, let $a$ and $b$ be lines
  such that $(a,b)$ is the topmost good pair in the sequence for which
  a line $c$ exists on $e''$ that crosses $a$ but not $b$. If $(a,b)$ is below
  $p'$, then the algorithm would reuse $p'$ instead of the new pair
  $p''$, since $(a,b)$ is in a sequence below $p$; hence, also
  $p'$ is in the sequence and above $(a,b)$.

  Now suppose that $(a,b)$ is above $p'$. The pair $(a,b)$ is created
  by inheritance because $c$ ends between $a$ and $b$. As both
  $a$ and $b$ end to the top, separated from the bottom side of the
  path by $p'$, this inheritance takes place at a node, where
  $a$ is the last line to end on the top side. But in this case
  $c$ prevents the inheritance of the good pair $(a,b)$ because it
  crosses only $a$, a contradiction.
\end{proof}

\section{Block Crossings on Trees}
\label{sec:tree}

In what follows we focus on instances of (M)BCM that are
trees. We first give an algorithm that bounds the number of block
crossings. Then, we consider trees with an additional constraint on the
lines; for these we develop a 6-approximation for MBCM.


\begin{theorem}
\label{thm:tree}
For any tree $T$ and lines $L$ on $T$, we can order
the lines with at most $2|L|-3$ \emph{monotone} block crossings in
$O(|L|  (|L|+|E|))$ time.
\end{theorem}
\emph{Proof.}
We give an algorithm in which paths are inserted one by one into the current
order; for each newly inserted path we create at most 2 monotone block crossings.
The first line cannot create a crossing, and the second line
crosses the first one at most once.

\begin{wrapfigure}{r}{0.29\textwidth}
  \centering
  \includegraphics[]{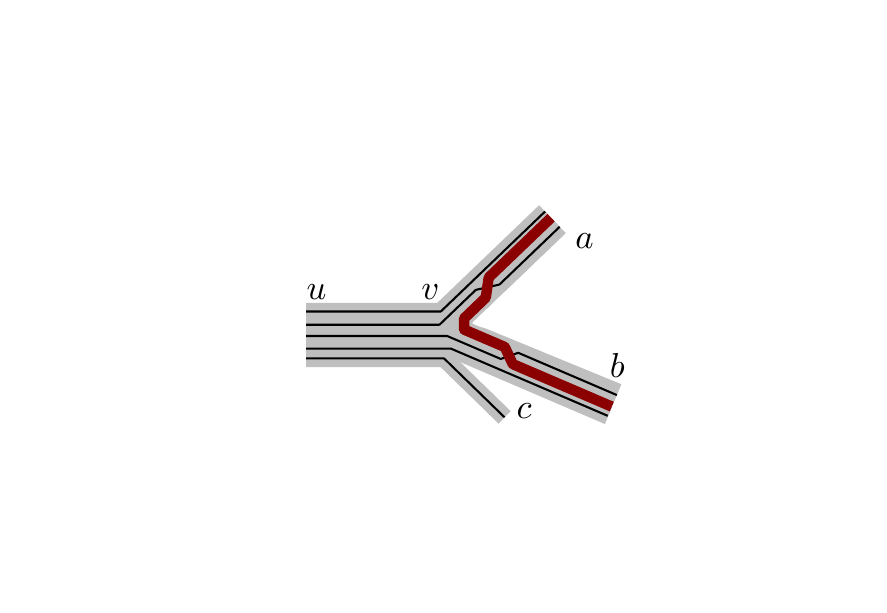}
  \caption{Insertion of a new line (red, fat) into the current order on edges $(v, a)$
  and $(v, b)$.}
  \label{fig:tree1}
\end{wrapfigure}
We start at an edge $(u, v)$ incident to a terminal.
When processing the edge the paths $L_{uv}$ are already in the correct order;
they do not need to cross on yet unprocessed edges of $T$. We consider all
unprocessed edges $(v, a),(v, b),\ldots$ incident to $v$ and build the
correct order for them.
The relative order of lines also passing through $(u, v)$ is kept unchanged.
For all lines passing through $v$ that were not treated before, we
apply an insertion procedure; see Fig.~\ref{fig:tree1}. Consider, e.g.,
the insertion of a line passing through $(v, a)$ and $(v, b)$. Close to $v$ we
add $l$ on both edges at the innermost position such that we do not
get vertex crossings with lines that pass through $(v, a)$ or $(v, b)$.
We find its correct position in the current order of lines $L_{va}$ close to
$a$,
and insert it using one block crossing.
This crossing will be the last one on $(v, a)$ going from $v$ to
$a$. Similarly, $l$ is inserted into $L_{vb}$.
We have to make sure that
lines that do not have to cross are inserted in the right order. As we
know the right relative order for a pair of such lines we can make
sure that the one that has to be innermost at node $v$ is inserted
first. Similarly, by looking at the clockwise order of edges around $v$,
we know the right order of line insertions such that there are no avoidable vertex crossings.
When all new paths are inserted the orders on $(v, a),(v, b),\dots$
are correct; we proceed by recursively processing these edges.

When inserting a line, we create at most $2$ block crossings, one
per edge of $l$ incident to $v$. After inserting the first two lines
into the drawing there is at most one crossing. Hence, we get at
most $2|L|-3$ block crossings in total. Suppose monotonicity would be
violated, that is, there is a pair of lines that crosses twice. The
crossings then have been introduced when inserting the second of those
lines on two edges incident to a node $v$. This can, however, not
happen, as at node $v$ the two edges are inserted in the right order.
Hence, the block crossings of the solution are monotone.\qed

In the following we show that the upper bound that our algorithm yields is tight.

\paragraph{Worst-Case Examples.}
\label{sec:worst-case}
Consider the graph show in Fig.~\ref{fig:worst-case-tree}.
The new green path in Fig.~\ref{fig:worst-case-tree}b) is inserted so
that it crosses 2 existing paths. This is the induction step for
creating instances in which $2|L|-3$ block crossings are necessary in
any solution.

\begin{figure}[bt]
    \centering
    \includegraphics{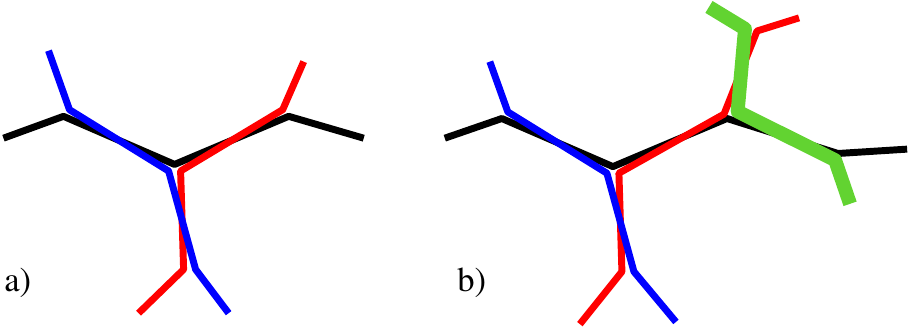}
    \caption{Tree with $2|L|-3$ necessary crossings for a) $|L|=3$, b)
    $|L|=4$.}
    \label{fig:worst-case-tree}
\end{figure}

We also have a simple example in which the tree algorithm creates
$|L|-1$ crossings while a single block crossing suffices; see
Fig.~\ref{fig:alg-tree-counterexample} for $|L|=5$. The example can easily be
extended to any number of lines. This shows that the algorithm does
not yield a constant factor approximation.
\begin{figure}[bt]
		\centering
    a)\includegraphics[page=1,width=0.45\textwidth]{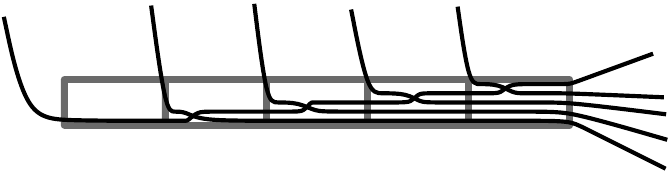}
    \hfill
    b)\includegraphics[page=2,width=0.45\textwidth]{pics/alg-tree-counterexample}
    \caption{a) The algorithm started leftmost produces
    $4$ crossings; b) one block crossing suffices.}
		\label{fig:alg-tree-counterexample}
\end{figure}

The examples shows that the algorithm described in Theorem~\ref{thm:tree} does
not guarantee an approximation of the optimal solution.
Next, we introduce an additional constraint on the lines, which helps us
to approximate the minimum number of block crossings.

\subsubsection{Upward Trees.}
\label{sec:utree}
We consider MBCM on an \emph{upward} tree~$T$, that is, a tree that has a
planar upward drawing in which all paths are monotone in vertical
direction, and all path sources are on the
same height as well as all path sinks; see Fig.~\ref{fig:utree}.
Bekos~et~al.~\cite{bekos08} already considered such trees (under the name
``left-to-right trees'') for the metro line crossing minimization problem.
Note that a graph whose skeleton is a path is not necessarily an upward tree.
Our algorithm consists of three steps. First, we perform a
simplification step removing some lines. Second, we use the algorithm for
trees given in  Sec.~\ref{sec:tree} on a simplified instance. Finally, we reinsert
the removed lines into the constructed order. We first analyze the upward
embedding.

\begin{figure}[t]
		\centering
    \includegraphics[width=0.9\textwidth]{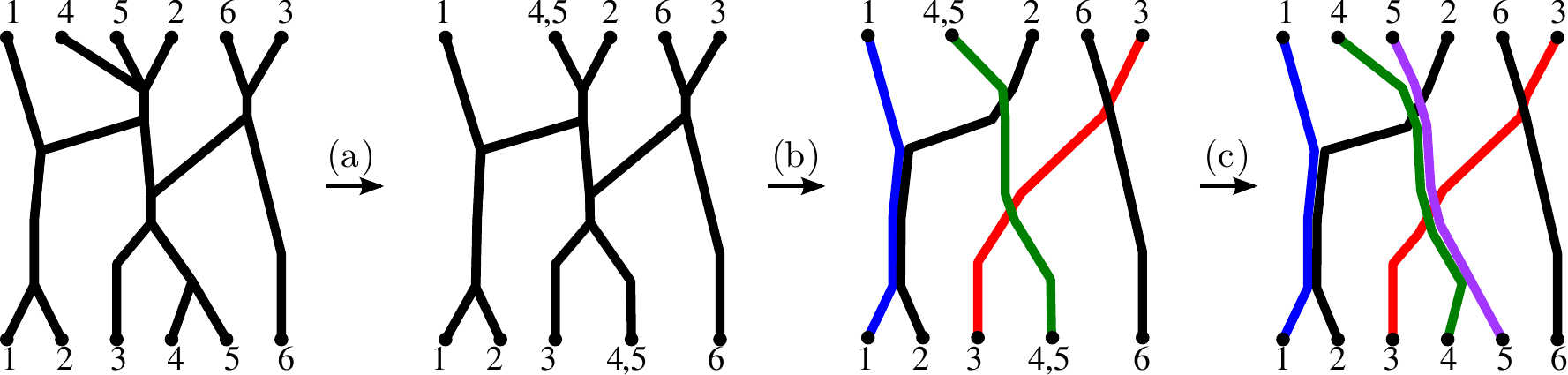}
    \caption{Algorithm for upward trees: (a) simplification, (b) line
    ordering, (c) reinsertion.}
    \label{fig:utree}
\end{figure}

Given an upward drawing of $T$, we read a permutation $\pi$
produced by the terminals on the top; we assume that the terminals
produce the identity permutation on the bottom.
Similar to the single edge case the goal is to sort $\pi$
by a shortest sequence of block moves. Edges of $T$ restrict some block
moves on $\pi$; e.g., the blocks $1\:4$ and $5$ in Fig.~\ref{fig:utree} cannot be
exchanged as there is no suitable edge. However, we can use the
lower bound for block crossings on a single edge, see
Sec.~\ref{sec:edge}: For sorting a simple permutation $\pi$, at least
$\bp(\pi)/3$ block moves are necessary.
We stress that simplicity of $\pi$ is crucial here. To get an approximation,
we show how to simplify a tree.

Consider two non-intersecting paths $a$ and $b$ that are adjacent in both permutations and
share a common edge. 
We prove that one of these paths can
be removed without changing the optimal number of block crossings. First,
if any other line $c$ crosses $a$ then it also crosses $b$ (i). This is implied
by planarity and $y$-monotonicity of the drawing. Second, if $c$ crosses
both $a$ and $b$ then all three paths share a common edge (ii); otherwise,
there would be a cycle due to planarity.
Hence, for any solution for the paths $L-\left\{ b \right\}$, we
can construct a solution for $L$ by inserting $b$ without any new
block crossing.
To insert $b$, we must first move all block crossings on $a$ to
the common subpath with $b$. This is possible due to
observation~(ii). Finally, we can place $b$ parallel to $a$.

To get a 6-approximation for an upward tree $T$, we first remove lines
until the tree is simple.
Then we apply the insertion algorithm presented in Sec.~\ref{sec:tree}, and finally reinsert
the lines removed in the first step. The number of block crossings is
at most $2|L'|$, where $L'$ is the set of lines of the simplified instance.
As an optimal solution has at least $|L'|/3$ block
crossings for this simple instance, and reinserting lines does not
create new block crossings, we get the following theorem.

\begin{theorem}
  The algorithm yields a 6-approximation for MBCM on upward
  trees.
  \label{thm:upward_tree}
\end{theorem}

\section{Block Crossings on General Graphs}
\label{sec:graph}
Finally, we consider general graphs. We suggest an algorithm that
achieves an upper bound on the number of block crossings and show that
it is asymptotically worst-case optimal. Our
algorithm uses monotone block moves, that is, each pair of lines
crosses at most once.
The algorithm works on any embedded graph; it does not even need to be
planar, we just need to know the circular order of incident edges
around each vertex.

\begin{algorithm}[b]
  \caption{Ordering the lines on a graph}
  \label{algo:gen-alg-edge-ordering}
  \ForEach{edge $e$ with $|L_{e}| > 1$}{
  Build order of lines on both sides of $e$\;
  Merge lines that are in the same group on both sides\;
  Find the largest group of consecutive lines that stay parallel on
  $e$\;
  Insert all other lines into this group and undo merging\;}
\end{algorithm}

\begin{figure}[tb]
  \centering
    \includegraphics[scale=1,page=2]{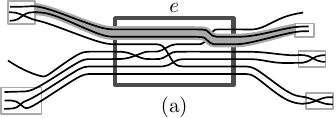}
		\hfill
    \includegraphics[scale=1,page=3]{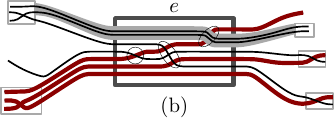}
  \caption{Sorting the lines on edge $e$. (a) Cutting edges (marked)
    define groups. The lines marked in gray are merged as they are in
    the same group on both sides. (b) Sorting by insertion into the largest
    group (red, fat); the merged lines always stay together, especially in their
block crossing.}
  \label{fig:gen-graph-edge}
\end{figure}

The idea of our algorithm is simple. We go through the edges in some
arbitrary order, similar to previous work on standard metro-line crossing
minimization~\cite{argyriou09,nollenburg09}. When we treat an edge, we completely
sort the lines that traverse it.  A crossing between a pair of lines can be
created on the edge only if this edge is the first one treated by the
algorithm that is used by both lines of the pair; see
Algorithm~\ref{algo:gen-alg-edge-ordering}.
The crucial part is sorting the lines on an edge.
Suppose we currently deal with edge $e$ and want to sort $L_e$.
Due to the path intersection property, the edge set used by the
lines in $L_e$ forms a tree on each side of $e$; see
Fig.~\ref{fig:gen-graph-edge}.
We cut these trees at those edges that have already been processed by
our algorithm.
Now, each line on $e$ starts at a leaf on one side and ends at a leaf on the
other side. Note that multiple lines can start or end at the same leaf.

From the tree structure and the orderings on the edges processed previously, we
get two orders of the lines, one on each side of $e$. We consider
\emph{groups of lines}
that start or end at a common leaf of the tree (like the red lines in
Fig.~\ref{fig:gen-graph-edge}). All lines of a
group have been seen on a common edge, and, hence, have been sorted.
Therefore lines of the same group form a consecutive subsequence on
one side of $e$, and have the same relative order on the other side of
$e$.

Let $g$ and $g'$ be a group of lines on the left and on the right side
of $e$, respectively. Suppose the set $L'$ of lines starting in
$g$ and ending in $g'$ consists of multiple lines. As the lines of
$g$ as well as the lines of $g'$ stay parallel on $e$, $L'$ must form
a consecutive subsequence (in the same order) on both sides. Now we
\emph{merge} $L'$ into one representative for the sequence of
lines, that is, we remove all lines of $L'$ and replace them by a
single line that is in the position of the lines of $L'$ on the
sequences on both sides of $e$. Once we find a
solution, we replace the
representative by the sequence without changing the number of block crossings.
Consider a crossing that involves the representative of~$L'$, that is, it is
part of one of the moved blocks. After replacing it, the sequence~$L'$
of parallel lines is completely contained in the same block. Hence, we do
not need additional block crossings.

We apply this merging for all pairs of groups on the left and right
end of $E$.  Then, we identify a group with the largest number of
lines after merging, and insert all remaining lines into it one by
one. Clearly, each insertion requires at most one block crossing; in
Fig.~\ref{fig:gen-graph-edge} we need three block crossings to insert
the lines into the largest (red) group.  After computing the
crossings, we undo the merging step and get a solution for edge $e$.

\begin{theorem}
  Algorithm~\ref{algo:gen-alg-edge-ordering} sorts all lines in
  $O(|E|^2|L|)$ time
  by monotone block moves. The resulting number of block crossings
  is $O(|L| \sqrt{|E'|})$, where
  $E'$ is the set of edges with at least two lines on them.
  \label{thm:alg-gen-graph}
\end{theorem}
\begin{proof}
  First, it is easy to see that no avoidable crossings are created, due to the
	path intersection property.
  Additionally, we care about all edges with at least two
  lines, which  ensures that all unavoidable crossings will be placed. Hence, we
  get a feasible solution using monotone crossings.
  Our algorithm sorts the lines on an edge in $O(|L||E|)$ time. We can
  build the tree structure and find the orders and groups by following
  all lines until we find a terminal or an edge that was processed
  before in $O(|L||E|)$ time. Merging lines and finding the largest
  group need $O(|L|)$ time; sorting by insertion into this group and
  undoing the merging can be done in $O(|L|^2)$ time. Note that
  $|L| \le |E|$ due to the path terminal property.

For analyzing the total number of block crossings, we maintain an
\emph{information} table~$T$ with $|L|^2$ entries. Initially, all the
entries are empty. After
processing an edge $e$ in our algorithm, we fill entries
$T[l,l']$ of the table for each pair $(l,l')$ of lines
that we see together for the first time. The main idea is that with
$b_{e}$ block crossings on edge $e$ we fill at least $b_{e}^2$ new
entries of $T$. The upper bound then can be concluded.

More precisely, let the \emph{information gain} $I(e)$
be the number of pairs of (not necessarily distinct) lines $l, l'$ that
we see together on a common edge $e$ for the first time. Clearly,
$\sum_{e \in E}I(e) \le |L|^{2}$. Suppose that $b_{e}^{2} \le I(e)$
for each edge~$e$. Then, $\sum_{e \in E} b_{e}^2 \le \sum_{e \in E}
I(e) \le |L|^{2}$. Using the Cauchy-Schwarz inequality $|\langle
x,y\rangle| \le \sqrt{\langle x,x\rangle\cdot \langle y,y\rangle}$ with $x$
as the vector of the $b_e$ and $y$ as a
vector of $1$-entries, we see that the total number of block crossings
is $\sum_{e \in E'} b_{e} \le |L| \sqrt{|E'|}$.

We still have to show that $b_e^2 \le I(e)$ for an edge $e$. For doing
so, we analyze the lines after the merging step.
Consider the groups on both sides of~$e$; we number the groups on
the left side $\mathfrak{L}_{1}, \dots, \mathfrak{L}_{n}$ and the groups on
the right side
$\mathfrak{R}_{1}, \ldots, \mathfrak{R}_{m}$ with $l_i =
|\mathfrak{L}_i|, r_j = |\mathfrak{R}_j|$ for $1\le i \le
n$, $1 \le j \le m$.  Without loss of generality, we can assume that
$\mathfrak{L}_1$ is the largest group into which all remaining lines are inserted.
Then, $b_e \le |L_e| - l_1$.
Let $s_{ij}$ be the number of lines that are in group $\mathfrak{L}_i$ on the left
side and in group $\mathfrak{R}_j$ on the right side of $e$. Note that $s_{ij} \in
\left\{ 0,1 \right\}$, otherwise we could still merge lines. Then $l_i =
\sum_{j} s_{ij}$, $r_j = \sum_{i}s_{ij}$, $s := |L_e| = \sum_{ij} s_{ij}$, and $b_e = s - l_1$. The information
gain is $I(e) = s^2 - \sum_{i} l_i^2 - \sum_j r_j^2 + \sum_{ij} s_{ij}^2$.
By applying Lemma~\ref{lemma:matrix-lemma-app}
we get $b_{e}^2 \le I(e)$. To complete the proof, note
that the unmerging step cannot decrease~$I(e)$.
\end{proof}

\begin{lemma}
\label{lemma:matrix-lemma-app}
  For $1 \le i \le n$ and $1 \le j \le m$ let
  $s_{ij} \in \left\{ 0,1 \right\}$.
  Let $l_{i} = \sum_{j}s_{ij}$ for $1 \le i \le n$
  and let $r_{j} = \sum_{i} s_{ij}$ for $1 \le j \le m$ such that
  $l_{1} \ge l_i$ for $1 \le i \le n$ and $l_{1} \ge r_j$ for $1 \le j
  \le m$. Let $s = \sum_{i=1}^n \sum_{j=1}^m s_{ij}$, $b = s - l_1$,
  and $I = s^2 - \sum_{i} l_i^2 - \sum_j r_j^2 + \sum_{ij} s_{ij}^2$.
  Then $b^2 \le I$.
\end{lemma}
\begin{proof}
  It is easy to check that for any $i,j$ it holds that
  $s_{ij}(s_{ij}-s_{1j}) \ge 0$ as $s_{ij} \in \left\{
  0,1 \right\}$. Using this property, we see that
  \begin{eqnarray*}
    I - b^2 &=& s^2 - \sum_{i=1}^n l_i^2 -
    \sum_{j=1}^m r_j^2 + \sum_{i=1}^n \sum_{j=1}^m s_{ij}^2 - s^2 + 2 s
    l_{1} - l_{1}^2\\
    &=& \sum_{i=1}^n \sum_{j=1}^m s_{ij}^2 + 2l_{1}(s - l_1) -
    \sum_{i=2}^n l_i^2  - \sum_{j=1}^m r_j\sum_{i=1}^n s_{ij} \\
    &=& \sum_{i=1}^n \sum_{j=1}^m s_{ij}^2 + 2l_{1} \sum_{i=2}^n
    \sum_{j=1}^m s_{ij} - \sum_{i=2}^n l_i \sum_{j=1}^m s_{ij}  -
    \sum_{i=1}^n \sum_{j=1}^m s_{ij}r_j \\
    &=& \sum_{i=2}^n \sum_{j=1}^m s_{ij} \left(
    s_{ij} + 2l_{1} - l_i - r_j \right) - \sum_{j=1}^m s_{1j}
    \sum_{i=2}^n s_{ij}\\
    & \ge & \sum_{i=2}^n \sum_{j=1}^m s_{ij} \left( s_{ij} - s_{1j}
  \right) \ge 0.
  \end{eqnarray*}
\end{proof}

We can show that the upper bound on the number of block crossings
that our algorithm achieves is tight. To this end, we use the
existence of special Steiner systems for building (non-planar) worst-case
examples of arbitrary size in which many block crossings are
necessary; see Theorem~\ref{thm:worst-case-general}.

\begin{theorem}
  There exists an infinite family of graphs $G = (V,E)$ with set of
  lines $L$ so that $\Omega(|L|\sqrt{|E'|})$ block crossings are necessary in
  any solution, where $E'$ is the set of edges with at least two lines on them.
  \label{thm:worst-case-general}
\end{theorem}
\begin{proof}
  From the area of projective planes it is known that, for any prime
  power $q$, a $S(q^2+q+1, q+1, 2)$ Steiner system
  exists~\cite{veblen1906finite}, that is,
  there is a set $\mathcal{S}$ of $q^2+q+1$ elements with subsets $S_{1}, S_2,
  \ldots, S_{q^2+q+1}$ of size $q+1$ such that any pair of elements
  $s, s' \in \mathcal{S}$ appears together in exactly one set $S_i$.

  We build a graph $G = (V,E)$ by first adding vertices $s_1$,
  $s_2$ and an edge $(s_{1}, s_{2})$ for any $s \in \mathcal{S}$. These edges
  will be the only ones with multiple lines on them, that is, $E'$.
  Additionally, we add an edge
  $s_{2}, s'_{1}$ for any pair $s,s' \in \mathcal{S}$. Next, we build a line
  $l_i$ for any set $S_i$ as follows. We take some arbitrary order
	$s, s', s'', \ldots, s^{(q)}$ on the elements of $S_i$, and build the path
	$s(l_i), s_{1}, s_{2}, s_{1}', s_{2}', s_{1}'', \ldots, s_{2}^{(q)}, t(l_i)$ with
  extra terminals $s(l_i)$ and $t(l_i)$ in which $l_i$ starts and
	ends, respectively; see Fig.~\ref{fig:worst-case-gen-graph-line}. As any pair of lines shares exactly one edge
  the
  path intersection property holds. We order the
  edges around vertices $s_1$ and $s_2$ so that all $q+1$ lines on the edge
  representing any $s \in \mathcal{S}$ have to cross by making sure that
	the orders of lines in $s_1$ and $s_{2}$ are exactly reversed; see
	Fig.~\ref{fig:worst-case-gen-graph-edge}. Then, $q$ block
  crossings are necessary on each edge, and, hence, $(q^2+q+1)q =
  \theta(q^3)$ block crossings in total.  On the other hand,
  $|L|\sqrt{|E'|} = (q^2+q+1)\sqrt{q^2+q+1} = \theta(q^3)$.
  Note that the graph $G$ is not planar.
\end{proof}

\begin{figure}[ht]
	\subfigure[Path $L_i$ is routed through the edges representing $s, s', s'',
	\ldots, s^{(q)}$.]{
	\label{fig:worst-case-gen-graph-line}
	\includegraphics{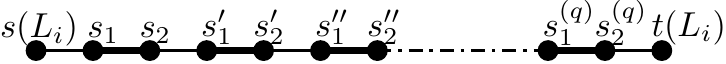}
	}
	\hfill
	\subfigure[The order of lines is reverted between $s_{1}$ and $s_{2}$.]{
	\label{fig:worst-case-gen-graph-edge}
	\includegraphics{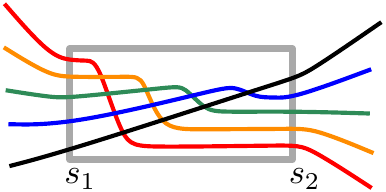}
	}
	\caption{Construction of the worst-case example.}
	\label{fig:wc-construction-gen-graph}
\end{figure}


\section{Conclusion and Open Problems}
\label{sec:conclusion}
We introduced a new variant of metro-line crossing minimization problem, 
and presented algorithms for single edges, paths, trees, and general graphs.
Our algorithm for general graphs cannot be applied if lines are more
complex subgraphs
than paths, or if the path intersection property does not hold.  On
the other hand, in many metro networks, there are just few lines violating
these properties. We suggest to first create an instance with
our properties by deleting few (parts of) lines. Then, after applying our
algorithm, the deleted parts can be reinserted by keeping them parallel to
other lines and reusing crossings as often as possible.

Another practical problem is the distribution of block crossings.
In our opinion, crossings of lines should preferably be close to the end
of their common subpath as this makes it easier to recognize that the
lines do cross. For making a metro line easy to follow the important
criterion is the number of its bends. Hence, an interesting question is how to
sort metro lines using the minimum total number of bends.

From a theoretical point of view, the complexity status of MBCM on a
single edge is an interesting open problem. Another question is whether there exist approximation algorithms for (M)BCM on trees or even general graphs.

\paragraph{Acknowledgments.}
We are grateful to Sergey Bereg, Alexander E. Holroyd, and Lev Nachmanson for
the initial discussion of the block crossing minimization problem, and for pointing out
a connection with sorting by transpositions.
We thank Jan-Henrik Haunert, Joachim Spoerhase, and Alexander Wolff
for fruitful discussions and suggestions.

\bibliographystyle{splncs03}
\bibliography{notes}

\begin{thebibliography}{10}
\providecommand{\url}[1]{\texttt{#1}}
\providecommand{\urlprefix}{URL }

\bibitem{argyriou09}
Argyriou, E.N., Bekos, M.A., Kaufmann, M., Symvonis, A.: On metro-line crossing
  minimization. J. Graph Algorithms Appl.  14(1),  75--96 (2010)

\bibitem{bafna98}
Bafna, V., Pevzner, P.A.: Sorting by transpositions. SIAM J. Discrete Math.
  11(2),  224--240 (1998)

\bibitem{bekos08}
Bekos, M.A., Kaufmann, M., Potika, K., Symvonis, A.: Line crossing minimization
  on metro maps. In: Hong, S.H., Nishizeki, T., Quan, W. (eds.) GD 2007. LNCS,
  vol. 4875, pp. 231--242. Springer, Heidelberg (2008)

\bibitem{benkert07}
Benkert, M., N{\"o}llenburg, M., Uno, T., Wolff, A.: Minimizing intra-edge
  crossings in wiring diagrams and public transport maps. In: Kaufmann, M.,
  Wagner, D. (eds.) GD 2006. LNCS, vol. 4372, pp. 270--281. Springer,
  Heidelberg (2007)

\bibitem{bulteau11}
Bulteau, L., Fertin, G., Rusu, I.: Sorting by transpositions is difficult. SIAM
  J. Discr. Math.  26(3),  1148--1180 (2012)

\bibitem{christie01}
Christie, D.A., Irving, R.W.: Sorting strings by reversals and by
  transpositions. SIAM J. Discr. Math.  14(2),  193--206 (2001)

\bibitem{elias06}
Elias, I., Hartman, T.: A 1.375-approximation algorithm for sorting by
  transpositions. IEEE/ACM Trans. Comput. Biol. Bioinformatics  3(4),  369--379
  (2006)

\bibitem{cogr}
Fertin, G., Labarre, A., Rusu, I., Tannier, E., Vialette, S.: Combinatorics of
  Genome Rearrangements. The MIT Press (2009)

\bibitem{groeneveld89a}
Groeneveld, P.: Wire ordering for detailed routing. IEEE Des. Test  6(6),
  6--17 (1989)

\bibitem{Heath98}
Heath, L.S., Vergara, J.P.C.: Sorting by bounded block-moves. Discrete Applied
  Mathematics  88(1--3),  181--206 (1998)

\bibitem{mareksadowska95}
Marek-Sadowska, M., Sarrafzadeh, M.: The crossing distribution problem. IEEE
  Trans. CAD Integrated Circuits Syst.  14(4),  423--433 (1995)

\bibitem{nollenburg09}
N{\"o}llenburg, M.: An improved algorithm for the metro-line crossing
  minimization problem. In: Eppstein, D., Gansner, E.R. (eds.) GD 2009. LNCS,
  vol. 5849, pp. 381--392. Springer, Heidelberg (2010)

\bibitem{pupyrev11}
Pupyrev, S., Nachmanson, L., Bereg, S., Holroyd, A.E.: Edge routing with
  ordered bundles. In: van Kreveld, M.J., Speckmann, B. (eds.) GD 2011. LNCS,
  vol. 7034, pp. 136--147. Springer, Heidelberg (2012)

\bibitem{schreiber2002high}
Schreiber, F.: High quality visualization of biochemical pathways in {BioPath}.
  In Silico Biology  2(2),  59--73 (2002)

\bibitem{veblen1906finite}
Veblen, O., Bussey, W.H.: Finite projective geometries. Trans. AMS  7(2),
  241--259 (1906), \url{http://www.jstor.org/stable/1986438}

\end{thebibliography}

\end{document}